\documentclass[copyright,creativecommons]{eptcs}
\providecommand{\event}{QPL 2013}

\usepackage{amsmath,amssymb,amsfonts,latexsym,stmaryrd}
\usepackage{epsfig}
\usepackage{tikz,ifdraft}
\usetikzlibrary{arrows,positioning,matrix,backgrounds,calc,fit}
\usepackage{QED}

\newtheorem{theorem}{Theorem}[section]

\newtheorem{proposition}[theorem]{Proposition}
\newtheorem{conjecture}[theorem]{Conjecture}

\newtheorem{remark}[theorem]{Remark}

\newcommand{\plusminus}{\pm}

\newcommand{\beqa}{\begin{eqnarray*}}
\newcommand{\eeqa}{\end{eqnarray*}\par\noindent}

\newcommand{\ket}[1]{{|} #1\rangle}

\newcommand{\IMP}{\; \Rightarrow \;}

\title{A classification of multipartite states by degree of non-locality}
\author{Samson Abramsky
\institute{Department of Computer Science\\ University of Oxford}
\email{samson.abramsky@cs.ox.ac.uk}
\and Carmen Constantin
\institute{Department of Computer Science\\ University of Oxford}
\email{carmen.constantin@univ.ox.ac.uk}
}
\def\titlerunning{A classification of multipartite states by degree of non-locality}
\def\authorrunning{Samson Abramsky and Carmen Constantin}

\begin{document}

\maketitle

\begin{abstract}
We propose a novel form of classification of multipartite states, in terms of the maximum degree of non-locality they can exhibit under any choice of local observables. 
This uses the hierarchy of notions previously introduced by Abramsky and Brandenburger: strong contextuality, logical contextuality, and probabilistic contextuality. 

We study n-qubit pure states. 
We conjecture that for more than 2 parties, all entangled states are logically contextual. 
We prove a number of results in support of this conjecture: 
(1) We show that \emph{all permutation-symmetric states are logically non-local}. 
(2) We study the class of \emph{balanced states with functional dependencies}. These states are described by Boolean functions and have a rich structure, allowing a detailed analysis, which again confirms the conjecture in this case. 
\end{abstract}

\section{Introduction}

A general understanding of the structure of multipartite entangled quantum states has proved elusive.
The picture given by the \textsc{SLOCC} classification \cite{dur2000three,lamata2006inductive} does not yield much insight beyond the tripartite case. Thus it seems worthwhile to consider other approaches.

Our starting point is the hierarchy of empirical models established in \cite{abramsky2011unified}.
An empirical model is a probability table describing measurement outcomes, familiar from Bell-type theorems. Such a table can be realized in quantum mechanics by fixing a multipartite state, and a set of local observables at each site.

In \cite{abramsky2011unified}, a general approach was developed which unifies the study of non-locality and contextuality.
One of the key points which emerged from this analysis is that three grades or degrees of contextuality/non-locality for empirical models can be distinguished, and shown to form a strict hierarchy:
\begin{itemize}
\item A model is \emph{strongly contextual} if its support has no global section; that is, there is no simultaneous assignment of outcomes to all the measurements whose restriction to each compatible set of measurements is in the support. Strong contextuality has a number of equivalent characterizations. In quantum mechanics, the model generated by the GHZ state with $X$ and $Y$ local observables provides a standard example of strong contextuality; while PR boxes \cite{popescu1994quantum} are super-quantum devices exhibiting strong contextuality.
\item A model is \emph{logically} (or \emph{relationally} or \emph{possibilistically}) \emph{contextual} if the following holds.
Let $U_1, \ldots , U_N$ be the compatible families of measurements.
Let $S_i$ be the support of the model on the joint outcomes for $U_i$.
Then for some $j$, there is a proper subset $S$ of $S_j$ such that the set of global sections which are compatible with $S_i$, $i \neq j$, all restrict to $S$ at $U_j$.
This says that there are events in the support at $U_j$ which are not consistent with the supports of the other measurement contexts, when viewed as constraints on a putative global section or hidden variable.
This condition captures in a precise way the idea of giving a proof of Bell's theorem without inequalities or probabilities \cite{greenberger1990bell,mermin1990quantum,cabello2001bell,hardy1993nonlocality,zimba1993bell}. It is characteristic of the well-known Hardy construction \cite{hardy1993nonlocality}, which is logically but not strongly contextual.
\item Finally, a model is \emph{weakly contextual} if it is contextual, but neither strongly nor logically contextual.
\end{itemize}
These notions form a proper hierarchy. 
In our setting, non-locality is a special case of contextuality.
Strong contextuality implies logical contextuality, which implies contextuality in the usual sense. 
There are weakly contextual models which are not logically contextual, and logically contextual models which are not strongly contextual.

We now turn to the issue of classifying quantum states in terms of their non-locality properties.\footnote{Since we are in the case of Bell-type scenarios as standardly discussed in non-locality theory, we shall use the terminology of non-locality rather than contextuality.}
In particular, we shall focus on $n$-qubit pure  states. If we fix local observables for each party, such a state gives rise to a probability model as above.
We can lift the properties of models to states. 
\begin{itemize}
\item We say that a state is strongly non-local if for \emph{some} choice of local observables for each party, the resulting empirical model is strongly contextual.
\item We can similarly define logical non-locality for states; we say that a state is logically non-local if for some choice of local observables, the resulting empirical model is logically non-local; while the state is \emph{not} strongly non-local.
\item Finally, a state is probabilistically non-local if it is non-local, but neither of the previous two cases apply.
\end{itemize}

This gives rise to a natural and challenging problem:
\begin{center}
\fbox{Characterize the multipartite states in terms of their maximum degree of non-locality.}
\end{center}
We believe that an answer to this problem will shed considerable light on the structure of multipartite states, not least because it will necessitate solving the following task:
\begin{center}
\fbox{Given a multipartite state, find local observables which witness its highest degree of non-locality.}
\end{center}

This problem motivates the following
\begin{conjecture}
\label{conj}
For every $n > 2$, every $n$-partite entangled state is logically non-local.
\end{conjecture}
Part of the thinking behind this conjecture is that the bipartite case may actually be anomalous within the landscape of multipartite entangled states. For example, the only strongly contextual bipartite models are the PR-boxes, which are of course not quantum realizable. By contrast, for all $n>2$, the $n$-partite GHZ states are strongly contextual \cite{abramsky2011unified}. Moreover, it is known that in the bipartite case, all entangled states \emph{except} the maximally entangled ones admit Hardy arguments, and hence are logically contextual \cite{hardy1993nonlocality}; and it seems to be folklore that this holds generally, and that a Hardy-type argument requires some symmetry-breaking.
However, as we shall see in the next section, for $n>2$ a different picture emerges.

In the remainder of this extended abstract, we shall report on progress towards proving the conjecture:
\begin{itemize}
\item In Section~2, we shall show that \emph{all permutation-symmetric states are logically non-local}.
This makes use of results from \cite{wang2012nonlocality}, which imply that all non-maximally permutation-symmetric states are logically contextual, combined with a direct argument to show that the Dicke states \cite{dicke1954coherence}, the maximally permutation symmetric states, are logically contextual with respect to $X$ and $Z$ local observables.
\item In Section~3, we consider a class of highly non-permutation-symmetric entangled states, the \emph{balanced states with functional dependencies}. These states are described by Boolean functions, and have a rich structure, allowing a detailed analysis, which again confirms the conjecture in this case.
\end{itemize}

\section{Permutation-symmetric states}

A permutation-symmetric $n$-qubit state is one which is invariant under the action of the full symmetry group $S_n$.
A natural basis for the permutation-symmetric states is provided by the  \emph{Dicke states} \cite{dicke1954coherence}, which are also physically significant. 
For each $n \geq 2$, $0 < k < n$ we define:
\[ S(n, k) \; := \; K \sum_{\mbox{{\small perm}}} \ket{0^{k}1^{n-k}} . \]
Here $K = {n \choose k}^{-1/2}$ is a normalization constant, and we sum over all products of $k$ $0$-kets and $n-k$ $1$-kets.

The well-known $W$ state is the $S(3,2)$ Dicke state in the above notation. 

\begin{proposition}
\label{Dickeprop}
For each $n > 2$, and $0 < k < n$, the Dicke state $S(n,k)$ is logically non-local.
\end{proposition}
\begin{proof}
Note that we exclude the cases $k=0$ and $k=n$, since in these cases $S(n, k)$ = $\ket{0^n}$ or $\ket{1^n}$, and these are obviously product states. We also exclude the bipartite case, for which $S(2,1)$ is the EPR state $\frac{\ket{01} + \ket{10}}{\sqrt{2}}$. The bipartite case seems anomalous in a number of respects.

We shall fix the observables $X$ and $Z$ in each party.

A Dicke state $S(n, k)$ gives rise to an $(n, 2, 2)$ probability model, with a choice of two dichotomic observables, $X$ or $Z$, at each site. This table has $2^n$ rows, corresponding to the possible choices of an observable at each site. We shall focus firstly on the $\frac{n(n-1)}{2}$ rows $r_{ij}$, where $X$ observables are selected at sites $i$ and $j$, and $Z$ observables at the remaining sites. Let $S_{ij}$ be the support of the model at row $r_{ij}$.

Now consider any joint outcome $s$ for this row in which there are $k$ $+$-outcomes and $(n-k)$ $-$-outcomes, and the outcome for $X^i$ is different to the outcome for $X^j$. We claim that $s$ is not in $S_{ij}$. If we compute the inner product whose squared norm gives the probability for $s$, we see that there are two terms, of the form $+1/c$ and $-1/c$ respectively. Thus the probability of $s$ is $0$, and it is not in the support.
We can express this in logical terms by saying that $S_{ij}$ satisfies the formula
\begin{equation}
\label{XXZeq}
\bigwedge_{k \neq i,j, s(k) = {+}} z_k \; \wedge \; \bigwedge_{k \neq i,j, s(k) = {-}} \neg z_k \;\; \IMP \; \; (x_i \leftrightarrow x_j) . 
\end{equation}
We now consider the row where $Z$ measurements are selected by every party. The support of this row is described by the formula
\begin{equation}
\label{Zeq}
\bigvee_{\pi} \; [\bigwedge_{i=1}^k z_{\pi(i)} \; \wedge \; \bigwedge_{j=k+1}^n \neg z_{\pi(j)} ]. 
\end{equation}
This is the logical counterpart of the description of $S(n, k)$ in the $Z$-basis.

From each disjunct $D$ of~(\ref{Zeq}) together with the relevant instances of~(\ref{XXZeq}), we can prove that $x_i \leftrightarrow x_j$ for all $i$, $j$ such that $z_i$ and $z_j$ appear with opposite polarity in $D$.
Note that, by the conditions on $k$ and $n$, both polarities do appear in $D$. By the transitivity of logical equivalence, it follows that $x_i \leftrightarrow x_j$ can be derived for all $i$, $j$. Thus $D$, together with the formulas~(\ref{XXZeq}), implies the formula
\begin{equation}
\label{Xeq}
\bigwedge_{i,j} \; x_i \leftrightarrow x_j . 
\end{equation}
Thus~(\ref{Zeq}) together with the conjunction of all instances of~(\ref{XXZeq}) implies~(\ref{Xeq}).

It follows that any global section which satisfies these formulas must restrict to just two joint outcomes in the row where $X$ measurements are selected by every party, namely those with the same outcome at every part.

To complete the argument, it suffices to show that these two outcomes form a proper subset of the support at that row. If we calculate the probability for each of these events, we obtain
\[ \left(\frac{{n \choose k}}{(\sqrt{2})^n \sqrt{{n \choose k}}} \right)^2 \;\;\; = \;\;\; \frac{{n \choose k}}{2^n} . \]
Thus we must show that
\[ \frac{{n \choose k}}{2^n} \; < \; \frac{1}{2} , \]
or equivalently
\[ {n \choose k} \;\; < \;\; 2^{n-1} \;\; = \;\; \sum_{l=0}^{n-1} {n-1 \choose l} \]
which follows from Pascal's rule:
\[ {n \choose k} \; = \; {n-1 \choose k-1} + {n-1 \choose k} . \]
Note however that to obtain a strict inequality, we need the assumption that $n>2$; the argument for the EPR state $S(2,1)$ fails at exactly this point.
\end{proof}

Using the results of \cite{abramsky2012logical}, we automatically obtain a logical Bell inequality which is violated by $S(n, k)$; the violation is
\[ 1 \; - \; \frac{{n \choose k}}{2^{n-1}} . \]

We also note that logical non-locality is preserved by the action of local unitaries $U_1 \otimes \cdots \otimes U_n$. If a state $\ket{\psi}$ is logically non-local with respect to measurement bases 
\[ \eta_1^{+}, \eta_1^{-}, \ldots , \eta_n^{+}, \eta_n^{-} , \]
then $U_1 \otimes \cdots \otimes U_n \ket{\psi}$ is logically non-local with respect to the measurement bases
\[ U_1 \eta_1^{+}, U_1 \eta_1^{-}, \ldots , U_n \eta_n^{+}, U_n \eta_n^{-} . \]
This follows since inner products and hence probabilities are preserved:
\[ \begin{array}{lcl}
\langle U_1 \eta_1^{\plusminus}\otimes \cdots \otimes U_n \eta_n^{\plusminus} \mid (U_1 \otimes \cdots \otimes U_n) \ket{\psi} \rangle  & = & \langle (U_1  \otimes \cdots \otimes U_n) \eta_1^{\plusminus} \otimes \cdots \otimes \eta_n^{\plusminus} \mid (U_1 \otimes \cdots \otimes U_n) \ket{\psi} \rangle \\
& = & \langle (U_1 \otimes \cdots \otimes U_n)^{\dagger} (U_1  \otimes \cdots \otimes U_n) \eta_1^{\plusminus} \otimes \cdots \otimes \eta_n^{\plusminus}   \ket{\psi}  \\
& = & \langle  \eta_1^{\plusminus} \otimes \cdots \otimes  \eta_n^{\plusminus} \ket{\psi}  .
\end{array} \]
Thus the orbits of the Dicke states under the actions of local unitaries are all logically non-local.

\begin{theorem}
All permutation-symmetric $n$-partite entangled states, for $n>2$, are logically non-local.
\end{theorem}
\begin{proof}
In \cite{wang2012nonlocality} it is shown that all permutation-symmetric states \emph{except} the unitary orbit of the Dicke states admit a Hardy argument, making use of the Majorana representation of permutation-symmetric states.
This is easily converted into a proof of logical non-locality.
The theorem now follows by combining this result with Proposition~\ref{Dickeprop}.
\end{proof}

\section{Functionally dependent states}

We now turn to a class of highly non-permutation-symmetric states.

For each $n\geq 2$, a $n$-ary Boolean function is a function $F:\{0,1\}^n\rightarrow \{0,1\}$. Each $n$-ary Boolean function can be expressed as a multivariate polynomial over $GL(2)$:
$$F(x_1,\ldots,x_n)=a_0 + \sum_{i}a_1^{i}x_i+ \sum_{i,j}a_2^{i,j}x_ix_j+\ldots + a_n^{1,2,\ldots,n}x_1x_2\ldots x_n$$
There are $2^n=1 + n + {n \choose 2} + \ldots + {n\choose n}$ summands in the expression of the above polynomial, each of which containing a binary coefficient $a_t^{i_1,\ldots,i_t}$. Hence there are $2^{2^{n}}$ distinct $n$-variate polynomials over $GF(2)$. Alternatively, each $n$-ary Boolean function can be expressed as a propositional formula in the Boolean variables $x_1,\ldots,x_n$ \cite{hazewinkel2001encyclopaedia}.

We define a balanced $n+1$-qubit quantum state with a functional dependency given by a $n$-variate polynomial $F$ as above to be a state which has the form
$$\Psi_F(n+1)=\frac{1}{\sqrt{2^n}}\sum_{q_1q_2\ldots q_n=00\ldots0}^{11\ldots 1} |q_1q_2\ldots q_n F(q1,q2,\ldots,q_n) \rangle $$
when expressed in the $Z$-basis. 

In the rest of this section we shall classify the balanced functionally dependent $n+1$-qubit quantum states in terms of their contextuality properties. We shall do this first for the three-partite case. A classification of the $n+1$-qubit states for $n>2$ can then be obtained using the results from the three-partite scenarios.

\subsection{The three-partite case}

\subsubsection{Polynomials of degree zero}

There are $2^{2^{2}}=16$ three-partite balanced states with a functional dependency. Two of these, namely 
$$\frac{1}{2}|000\rangle+|010\rangle+|100\rangle+|110\rangle=\left(\frac{|0\rangle+|1\rangle}{\sqrt{2}}\right)^{\otimes2}\otimes|0\rangle$$
 and 
$$\frac{1}{2}|001\rangle+|011\rangle+|101\rangle+|111\rangle=\left(\frac{|0\rangle+|1\rangle}{\sqrt{2}}\right)^{\otimes2}\otimes|1\rangle$$ 
are obviously product states, and hence non-contextual. They correspond to the constant polynomials $F_0(q_1,q_2)=0$ and $F_1(q_1,q_2)=1$ respectively.

\subsubsection{Degree one polynomials}
There are six states whose corresponding polynomials have degree one. Two of these are given by the functional dependencies which correspond to the two-variable propositional formulas $XOR$ and $NXOR$. Another four states are given by so-called dictatorships, i.e. the value of the last qubit is dictated either by the value of the first qubit or by the value of the second qubit. We shall look at these two classes of states below.

\subsubsection{XOR and NXOR}
The polynomials corresponding to the $XOR$ and $NXOR$ states have the form $F^a_{XOR}(q_1,q_2)=a+q_1+q_2$ with $a=0$ for $XOR$ and $a=1$ for $NXOR$. 

\begin{theorem}
 The $XOR$ state is strongly contextual if each party chooses between $Y$ and $Z$ measurements.
\end{theorem}

\begin{proof}
The support of the probability table for the $XOR$ state is 
\begin{center}
\begin{tabular}{l|cccccccc}
& $+++$ & $++-$ & $+-+$ & $+--$ & $-++$ & $-+-$ & $--+$ & $---$ \\ \hline
$YYY$ &  $1$ & $1$ & $1$ & $1$ & $1$ & $1$ & $1$ & $1$ \\
$YYZ$ &  $0$ & $1$ & $1$ & $0$ & $1$ & $0$ & $0$ & $1$ \\
$YZY$ &  $0$ & $1$ & $1$ & $0$ & $1$ & $0$ & $0$ & $1$ \\
$ZYY$ &  $0$ & $1$ & $1$ & $0$ & $1$ & $0$ & $0$ & $1$ \\
$YZZ$ &  $1$ & $1$ & $1$ & $1$ & $1$ & $1$ & $1$ & $1$ \\
$ZYZ$ &  $1$ & $1$ & $1$ & $1$ & $1$ & $1$ & $1$ & $1$  \\
$ZZY$ &  $1$ & $1$ & $1$ & $1$ & $1$ & $1$ & $1$ & $1$ \\
$ZZZ$ & $1$ & $0$ & $0$ & $1$ & $0$ & $1$ & $1$ & $0$ \\
\end{tabular}
\end{center}

One can simply inspect the table above and check that none of the sections in the support of the $ZZZ$ row can be extended to global sections (i.e. each possible global  assignment consistent with the support of the $ZZZ$ row will restrict to a section outside the support on at least one of the three rows $YYZ$, $YZY$ and $ZYY$). Thus there cannot be any global assignment of outcomes whose restriction to each set of compatible measurements is in the support of the model. 

It is worth at this point to give a more formal expression to this argument in order to gain a better understanding of what is actually going on. For this recall that the $+$ and $-$ eigenstates of the $Z$ observable are $|0\rangle$ and $|1\rangle$ respectively while for the $Y$ observable they are (modulo some normalization constant which does not play any role in our argument) $|y_+\rangle:=|0\rangle + i|1\rangle$ and $|y_-\rangle:=|0\rangle-i|1\rangle$ respectively

We start our argument by assuming that a global section does exist. Assume next that this global section makes the assignment $z_3=+$. The probability of obtaining the outcome $z_1z_2+$ with $z_i\in\{+,-\}$ is given by the squared norm of the inner product 
$$\langle e_{z_1} e_{z_2} 0|XOR\rangle=\langle e_{z_1} e_{z_2}0|\frac{|000\rangle+|011\rangle+|101\rangle+|110\rangle}{2}$$
where $e_+=0$ and $e_-=1$. If we regard each $e_{z_i}$ as an element of $GF(2)$ then the inner product above is non-zero only if $$F^0_{XOR}(e_{z_1},e_{z_2})=e_{z_1}+e_{z_2}=0$$

So the sections in the support of the $ZZZ$ for which $z_3= +$ must have $z_1=z_2$, as the table confirms. 

Next consider the $YYZ$ set of compatible measurements. The probability (modulo normalization constants) of obtaining the outcome $y_1y_2+$ with $y_i\in\{+,-\}$ for this set of measurements is given by the squared norm of the inner product
\begin{equation}\label{yyz}
\langle y_{y_1} y_{y_2} 0|XOR\rangle =\langle y_{y_1} y_{y_2} 0|\frac{|000\rangle+|011\rangle+|101\rangle+|110\rangle}{2}
\end{equation}
We have
\begin{align*}
 \langle y_+y_+|&=\langle 00|+i\langle 01|+i\langle 10|-\langle11|\\
\langle y_+y_-|&=\langle 00|-i\langle 01|+i\langle 10|+\langle11|\\
\langle y_-y_+|&=\langle 00|+i\langle 01|-i\langle 10|+\langle11|\\
\langle y_-y_-|&=\langle 00|-i\langle 01|-i\langle 10|-\langle11| 
 \end{align*}
and since $F^0_{XOR}(0,1)=F^0_{XOR}(1,0)\neq 0$ the imaginary part of the tensor products above will not bring any contribution towards the value of the inner product (\ref{yyz}). The only contribution will come from the real part of the tensor products above, and it is easy to see that the inner product (\ref{yyz}) will vanish when $y_1=y_2$. So we must have $y_1\neq y_2$ in any global assignment which sends $z_3$ to $+$ in order to stay within the support of the $YYZ$ row.

On the other hand, the probabilities of obtaining the outcomes $y_1zy_3$ and $zy_2y_3$, where $z=z_1=z_2$, for the $YZY$ and $ZYY$ sets of compatible measurements are given by the inner products
 \begin{align}
  \langle y_{y_1} e_z y_{y_3} &|XOR\rangle =\left(\langle 0e_z0|+iy_3\langle 0e_z1|+iy_1\langle 1e_z0|-(y_1y_3)\langle 1e_z1|\right)\ |XOR\rangle\\ 
\langle  e_zy_{y_2}y_{y_3} &|XOR\rangle =\left(\langle e_z00|+iy_3\langle e_z01|+iy_2\langle e_z10|-(y_2y_3)\langle e_z11|\right)\ |XOR\rangle \label{zyy}
 \end{align}
If $e_z=0$ the imaginary part of the two expressions in (\ref{zyy}) will be equal to zero for all values of $y_i$. If $e_z=1$ the real part of the two expressions in (\ref{zyy}) will vanish for all values of $y_i$. In the first case the expressions are non-zero only if $y_1=y_2=-y_3$ and in the second case they are non-zero only if $y_1=y_2=y_3$. But both these assignments violate the previous requirement that $y_1\neq y_2$.

So far we have established the fact that no global section can assign the outcome $+$ to $z_3$. If on the other hand the outcome $-$ is assigned to $z_3$, we can construct a similar argument which yields a contradiction. This time the sections in the support of $ZZZ$ for which $z_3=-$ must have $z_1=-z_2$. The sections in the support of $YYZ$ must have $y_1=y_2$, while those in the support of $YZY$ and $ZYY$ must either have $y_1=-y_3=-y_2$ for $e_{z_2}=0$, $e_{z_3}=1$ or $y_1=y_3=-y_2$ for $e_{z_2}=1$ and $e_{z_3}=0$. \qed

\end{proof}

\begin{theorem}
 The $NXOR$ state is also strongly contextual if each party chooses between $Y$ and $Z$ measurements.
\end{theorem}

\begin{proof}
The support of the probability table for the $NXOR$ state is 
\begin{center}
\begin{tabular}{l|cccccccc}
& $+++$ & $++-$ & $+-+$ & $+--$ & $-++$ & $-+-$ & $--+$ & $---$ \\ \hline
$YYY$ &  $1$ & $1$ & $1$ & $1$ & $1$ & $1$ & $1$ & $1$ \\
$YYZ$ &  $1$ & $0$ & $0$ & $1$ & $0$ & $1$ & $1$ & $0$ \\
$YZY$ &  $1$ & $0$ & $0$ & $1$ & $0$ & $1$ & $1$ & $0$ \\
$ZYY$ &  $1$ & $0$ & $0$ & $1$ & $0$ & $1$ & $1$ & $0$ \\
$YZZ$ &  $1$ & $1$ & $1$ & $1$ & $1$ & $1$ & $1$ & $1$ \\
$ZYZ$ &  $1$ & $1$ & $1$ & $1$ & $1$ & $1$ & $1$ & $1$  \\
$ZZY$ &  $1$ & $1$ & $1$ & $1$ & $1$ & $1$ & $1$ & $1$ \\
$ZZZ$ & $0$ & $1$ & $1$ & $0$ & $1$ & $0$ & $0$ & $1$ \\
\end{tabular}
\end{center}

The argument for strong contextuality follows the same pattern as for the $XOR$ state. We assume by contradiction that a global section exists, and that it makes the assignment $z_3=+$. Then from the $ZZZ$ row we obtain the requirement that $z_1\neq z_2$. From the $YYZ$ row we obtain that $y_1=y_2$ and from the $YZY$ and $ZYY$ rows we obtain that $y_1\neq y_2$, which is a contradiction. 

Similarly, if $z_3=-$ we must have $z_1=z_2$ and $y_1\neq y_2$ from the $ZZZ$ and $YYZ$ rows. This means we must also have $y_1=y_2$ from the $YZY$ and $ZYY$ rows, which again is a contradiction.

Note at this point that the similarity between these two arguments for strong contextuality is due to the similar structure of the tables for the $XOR$ an $NXOR$ states. Namely, the second table can be obtained from the first by interchanging the $+$ and $-$ signs which label the table columns. Thus the second argument is the same as the first, only with the $+$ and $-$ signs interchanged.
 \qed

\end{proof}

\subsubsection{Dictatorships}

The four degree one polynomials of the form $F^a_{1}(q_1,q_2)=a+q_1$ and $F^a_{1}(q_1,q_2)=a+q_2$ where $a\in\{0,1\}$ correspond to the so-called dictatorship states, where the value of the last qubit is dictated by the value of either the first or of the second qubit. In the $Z$ basis these states are
$$\Delta^+_2:=\frac{|0\rangle+|1\rangle}{\sqrt{2}}\otimes \frac{|00\rangle+|11\rangle}{\sqrt{2}}$$
or
$$\Delta^-_2:= \frac{|0\rangle+|1\rangle}{\sqrt{2}}\otimes \frac{|01\rangle+|10\rangle}{\sqrt{2}}$$
if the dictatorship is given by the second qubit. Similarly, if the dictatorship is given by the first qubit, we have two possible states
$$\Delta^+_1:= \frac{|0_2\rangle+|1_2\rangle}{\sqrt{2}}\otimes \frac{|0_10_3\rangle+|1_11_3\rangle}{\sqrt{2}}$$
and
$$\Delta^-_1:=\frac{|0_2\rangle+|1_2\rangle}{\sqrt{2}}\otimes \frac{|0_a1_3\rangle+|1_10_3\rangle}{\sqrt{2}}$$
where the subscripts $1$, $2$ and $3$ indicate whether the qubit belongs to the first, second or third party respectively.

\begin{proposition}\label{bell}
The four dictatorship states are weakly contextual for suitable dichotomic choices of measurements.
\end{proposition}

\begin{proof}

Consider the general form of an observable, given in terms of angles $\theta$ and $\phi$ on the Bloch sphere
$$U(\theta,\phi):=\left(\begin{array}{cc}
                          \cos\theta& e^{-i\phi}\sin\theta\\
			  e^{i\phi}\sin\theta&-\cos\theta
                         \end{array}\right)$$

We will use the fact that the bell basis states $\Phi^+=\frac{|00\rangle+|11\rangle}{\sqrt{2}}$ and $\Phi^-=\frac{|00\rangle+|11\rangle}{\sqrt{2}}$ are weakly contextual with respect to suitable choices of measurements.

It can be machine checked that the state $\Phi^+$ is weakly contextual if we allow each party to choose between the measurements $A:=U \left(\frac{\pi}{2},\frac{\pi}{8}\right)$ and $B:=U\left(\frac{\pi}{2},\frac{5\pi}{8}\right)$, while the state  $\Phi^-$ is weakly contextual if we allow each party to choose between the measurements $C:=U\left(\frac{\pi}{8},\frac{\pi}{2}\right)$ and $D:=U\left(\frac{5\pi}{8},\frac{\pi}{2}\right)$

In fact, it can also be machine checked that this choice of measurements gives a maximal violation of Bell inequalities for both states.%, if the $\pi/2$ angle is kept fixed.

The probability models of the dictatorship states can be obtained from the probability models of the states $\Phi^+$ and $\Phi^-$ in a straightforward way. Let $|+_A\rangle$ and $|-_A\rangle$ stand for the eigenstates of $A$ and $|+_B\rangle$ and $|-_B\rangle$ stand for the eigenstates of $B$. 

Define the two constants
\begin{align*}
a_+:&=\frac{1}{\sqrt{2}}(\langle+_A|0\rangle+\langle+_A|1\rangle)\\
a_-:&=\frac{1}{\sqrt{2}}(\langle-_A|0\rangle+\langle-_A|1\rangle) 
\end{align*}
Note that $a_++a_-=1$. 
Similarly, define the two constants
\begin{align*}
b_+:&=\frac{1}{\sqrt{2}}(\langle+_B|0\rangle+\langle+_B|1\rangle)\\
b_-:&=\frac{1}{\sqrt{2}}(\langle-_B|0\rangle+\langle-_B|1\rangle) 
\end{align*}
Up to two decimal points precision, the probability table of the $\Phi^+$ state for the observables $A$ and $B$ is 
\begin{center}
\begin{tabular}{l|cccc}
& $++$ & $+-$ & $-+$ & $--$ \\ \hline
$AA$ &  $0.43$ & $0.07$ & $0.07$ & $0.43$  \\
$AB$ &  $0.07$ & $0.43$ & $0.43$ & $0.07$  \\
$BA$ &  $0.07$ & $0.43$ & $0.43$ & $0.07$  \\
$BB$ &  $0.07$ & $0.43$ & $0.43$ & $0.07$  \\
\end{tabular}
\end{center}

The inner product formula for computing probabilities implies that the probability table of the dictatorship state $\Delta^+_2$ can be expressed in terms of the constants $a_+$, $a_-$, $b_+$ and $b_-$ and the probability table of $\Phi^+$:

\begin{center}
\begin{tabular}{l|cccc|cccc}
& $+++$ & $++-$ & $+-+$ & $+--$ & $-++$ & $-+-$ & $--+$ & $---$ \\ \hline
$AAA$ &  $0.43a_+$ & $0.07a_+$ & $0.07a_+$ & $0.43a_+$ & $0.43a_-$ & $0.07a_-$ & $0.07a_-$ & $0.43a_-$  \\
$AAB$ &  $0.07a_+$ & $0.43a_+$ & $0.43a_+$ & $0.07a_+$ & $0.07a_-$ & $0.43a_-$ & $0.43a_-$ & $0.07a_-$\\
$ABA$ &  $0.07a_+$ & $0.43a_+$ & $0.43a_+$ & $0.07a_+$ & $0.07a_-$ & $0.43a_-$ & $0.43a_-$ & $0.07a_-$ \\
$ABB$ &  $0.07a_+$ & $0.43a_+$ & $0.43a_+$ & $0.07a_+$ & $0.07a_-$ & $0.43a_-$ & $0.43a_-$ & $0.07a_-$ \\
\hline
$BAA$ &  $0.43b_+$ & $0.07b_+$ & $0.07b_+$ & $0.43b_+$& $0.43b_-$ & $0.07b_-$ & $0.07b_-$ & $0.43b_-$  \\
$BAB$ &  $0.07b_+$ & $0.43b_+$ & $0.43b_+$ & $0.07b_+$ & $0.07b_-$ & $0.43b_-$ & $0.43b_-$ & $0.07b_-$\\
$BBA$ &  $0.07b_+$ & $0.43b_+$ & $0.43b_+$ & $0.07b_+$ & $0.07b_-$ & $0.43b_-$ & $0.43b_-$ & $0.07b_-$\\
$BBB$ &  $0.07b_+$ & $0.43b_+$ & $0.43b_+$ & $0.07b_+$ & $0.07b_-$ & $0.43b_-$ & $0.43b_-$ & $0.07b_-$\\
\end{tabular}
\end{center}

Note also that the table of the dictatorship state $\Delta^+_1$ will have the same values as the one above, but the rows will be indexed in the order $AAA$, $AAB$, $BAA$, $BAB$, $ABA$, $ABB$, $BBA$, $BBB$, since the coefficients $a_{+/-}$ and $b_{+/-}$ come from the second qubit's contribution to the inner product.

It is now straightforward to deduce that the states $\Delta^+_1$ and $\Delta^+_2$ are indeed weakly contextual for the same choice of measurements for which the $\Phi^+$  state is weakly contextual, since any probability distribution on the set of global sections of one of these two dictatorship states would restrict to a probability distribution on the set of global sections of the $\Phi^+$ state.

Next note that up to two decimal points precision, the probability table of the $\Phi^-$ state for the observables $C$ and $D$ is
\begin{center}
\begin{tabular}{l|cccc}
& $++$ & $+-$ & $-+$ & $--$ \\ \hline
$AA$ &  $0.43$ & $0.07$ & $0.07$ & $0.43$  \\
$AB$ &  $0.07$ & $0.43$ & $0.43$ & $0.07$  \\
$BA$ &  $0.07$ & $0.43$ & $0.43$ & $0.07$  \\
$BB$ &  $0.07$ & $0.43$ & $0.43$ & $0.07$  \\
\end{tabular}
\end{center}

and the probability tables of the $\Delta^-_1$ and $\Delta^-_2$ dictatorship states can be expressed in terms of the table above and four suitably defined constants $c_{+/-}$ and $d_{+/-}$, so by analogy with the $\Delta^+_1$ and $\Delta^+_2$ case, these states will also be weakly contextual.\qed

\end{proof}

\begin{theorem}\label{bellth}
None of the four dictatorship states is logically contextual, for \textbf{any} dichotomic choice of measurements.
\end{theorem}

\begin{proof}
The relationship between probability tables discussed in Proposition \ref{bell} allows us to reduce the problem to the bi-partite scenario. Thus we seek to prove that neither of the two Bell basis states is logically contextual for any given choice of measurements. 

Let $A:=U(\theta_1,\phi_1)$ and $B:=U(\theta_2,\phi_2)$. Let $c$, $s$ and $f$ stand for $\cos\frac{\theta_1}{2}$, $\sin\frac{\theta_1}{2}$ and $e^{i\phi_1}$ respectively. Similarly, let $k$, $z$ and $v$ stand for $\cos\frac{\theta_2}{2}$, $\sin\frac{\theta_2}{2}$ and $e^{i\phi_2}$ respectively. Then the general form of the probability model of the $\Phi^+$ state is 

\begin{center}
\begin{tabular}{l|cccc}
& $++$ & $+-$ & $-+$ & $--$ \\ \hline
$AA$ &  $|c^2+f^2\cdot s^2|^2$ & $|cs-f^2\cdot cs|^2$ & $|cs-f^2\cdot cs|^2$ & $|s^2+f^2\cdot c^2|^2$  \\
$AB$ &  $|ck+fv\cdot sz|^2$ & $|cz-fv\cdot sk|^2$ & $|sk-fv\cdot cz|^2$ & $|sz+fv\cdot ck|^2$  \\
$BA$ &  $|ck+fv\cdot sz|^2$ & $|sk-fv\cdot cz|^2$ & $|cz-fv\cdot sk|^2$ & $|sz+fv\cdot ck|^2$  \\
$BB$ &  $|k^2+v^2\cdot z^2|^2$ & $|kz-v^2\cdot kz|^2$ & $|kz-v^2\cdot kz|^2$ & $|z^2+v^2\cdot k^2|^2$  \\
\end{tabular}
\end{center}

In most cases, all of the sections in the model of $\Phi^+$ will be in the support, in which case the state is clearly not logically contextual. However, for certain values of $c$, $f$, $v$ and $k$ (which may be chosen independently of each other) the entries of the table above may vanish, which will exclude certain sections from the support. It suffices therefore to check that the resulting possibilistic models are not logically contextual for any choices of $c$, $f$, $v$ and $k$ (and implicitly also of $s$ and $z$) which would allow one or more of the above table entries to vanish. We therefore need to consider each element in the powerset of the following set of conditions on $c$, $s$, $f$, $v$, $z$ and $k$:

$$\mathcal{C}:=\left\{c\vee k\in\{0,\pm 1\},~ f\vee v\in\{\pm 1,~ \pm i\},~ f=\pm \frac{1}{v},~ c=\pm s,~ k=\pm z,~ ck=\pm sz,~ cz=\pm sk\right\} $$

A computer can easily verify that no subset of the above set of conditions leads to a logically contextual probability model. 

Finally, using the relation between probability tables from Proposition \ref{bell}, we note that any global section of the model above can be easily extended to a global section of the corresponding dictatorship state model by adding the  assignment $+$ to the third party's outcome for the $A$ measurement, if $a_+\neq0$ and $-$ otherwise, and similarly for the third party's outcome corresponding to the $B$ measurement. We can therefore conclude that for all possible choices of measurements, the dictatorship states corresponding to $\Phi^+$ can not be logically contextual. 

For the $\Phi^-$ state note that the observables $C:=U(\phi_1,\theta_1)$ and $D:=U(\phi_2,\theta_2)$ will give the probability model
\begin{center}
\begin{tabular}{l|cccc}
& $++$ & $+-$ & $-+$ & $--$ \\ \hline
$AA$ &  $|cs-f^2\cdot cs|^2$ &  $|c^2+f^2\cdot s^2|^2$ & $|s^2+f^2\cdot c^2|^2$ &  $|cs-f^2\cdot cs|^2$   \\
$AB$ &  $|cz-fv\cdot sk|^2$ & $|ck+fv\cdot sz|^2$ & $|sz+fv\cdot ck|^2$  &  $|sk-fv\cdot cz|^2$  \\
$BA$ &  $|sk-fv\cdot cz|^2$ &  $|ck+fv\cdot sz|^2$  & $|sz+fv\cdot ck|^2$ & $|cz-fv\cdot sk|^2$ \\
$BB$ &  $|kz-v^2\cdot kz|^2$ &  $|k^2+v^2\cdot z^2|^2$ & $|z^2+v^2\cdot k^2|^2$ & $|kz-v^2\cdot kz|^2$   \\
\end{tabular}
\end{center}
where $c,k,s,z$ now take $\phi_i/2$ as arguments while $f$ and $v$ take $\theta_i$ as arguments. 

We can show that this model is also not logically contextual, using an argument completely analogous to the one used for the $\Phi^+$ state. Hence the dictatorship states corresponding to the $\Phi^-$ state are also not logically contextual. \qed

\end{proof}

\subsubsection{Degree two polynomials}\label{two}

There are eight balanced functionally dependent states whose corresponding polynomials have degree two. Four of these correspond to the two-variable propositional formulas $AND,$ $NAND,$ $OR$ and $NOR$. Their respective polynomials have the form $$F^a_{AND}(q_1,q_2)=a+q_1q_2$$ and $$F^a_{OR}=a+q_1+q_2+q_1q_2$$ with $a=0$ for $AND$ and $OR$ and $a=1$ for $NAND$ and $NOR$. 

The other four states correspond to logical implication and its negation. We use $L_1,\ L_{2},\ NL_{1}$ and $NL_{2}$ to denote the propositional formulas $q_1\Rightarrow q_2,\ q_2\Rightarrow q_1$ and $\overline{q_1\Rightarrow q_2}$, $\overline{q_2\Rightarrow q_1}$ respectively. The polynomials corresponding to these propositional formulas are of the form $$F^a_{NL_{i}}=a+q_i+q_1q_2$$ with $i\in\{1,2\}$, $a=0$ for $NL_i$ and $a=1$ for $L_i$.

All the eight states described above turn out to be logically contextual if we choose $Y$ and $Z$ measurements in each part.

\begin{theorem}\label{and}
 The $AND$ state is logically contextual. 
\end{theorem}

\begin{proof}
The support of the probability table for the $AND$ state is 
\begin{center}
\begin{tabular}{l|cccccccc}
& $+++$ & $++-$ & $+-+$ & $+--$ & $-++$ & $-+-$ & $--+$ & $---$ \\ \hline
$YYY$ &  $1$ & $1$ & $1$ & $1$ & $1$ & $1$ & $1$ & $1$ \\
$YYZ$ &  $1$ & $1$ & $1$ & $1$ & $1$ & $1$ & $1$ & $1$ \\
$YZY$ &  $1$ & $1$ & $0$ & $1$ & $1$ & $1$ & $1$ & $0$ \\
$ZYY$ &  $1$ & $1$ & $1$ & $1$ & $0$ & $1$ & $1$ & $0$ \\
$YZZ$ &  $1$ & $0$ & $1$ & $1$ & $1$ & $0$ & $1$ & $1$ \\
$ZYZ$ &  $1$ & $0$ & $1$ & $0$ & $1$ & $1$ & $1$ & $1$  \\
$ZZY$ &  $1$ & $1$ & $1$ & $1$ & $1$ & $1$ & $1$ & $1$ \\
$ZZZ$ & $1$ & $0$ & $1$ & $0$ & $1$ & $0$ & $0$ & $1$ \\
\end{tabular}
\end{center}

The global assignment $z_1z_2z_3y_1y_2y_3=++++++$ is clearly consistent with the support of the $AND$ table, so this state is not strongly contextual for $Y$ and $Z$ measurements. However, not all sections in the support can be extended to global sections. Consider for example the section $y_1y_2z_3=+--$ which is in the support. The only section on the $ZZZ$ row consistent with it is $z_1z_2z_3=---$. But it is now impossible to assign an outcome to $y_3$ which will make the resulting global section restrict to sections in the support of both of the rows $YZY$ and $ZYY$. In fact, there are only two sections in the support of the $YYZ$ row which cannot be extended to global ones. These are the sections where the two $Y$ measurements are assigned different outcomes, while the $Z$ measurement is assigned the outcome $-$. \qed
\end{proof}

\begin{theorem}\label{nand}
 The $NAND$ state is logically contextual.
\end{theorem}

\begin{proof}
The support of the probability table for the $NAND$ state is 
\begin{center}
\begin{tabular}{l|cccccccc}
& $+++$ & $++-$ & $+-+$ & $+--$ & $-++$ & $-+-$ & $--+$ & $---$ \\ \hline
$YYY$ &  $1$ & $1$ & $1$ & $1$ & $1$ & $1$ & $1$ & $1$ \\
$YYZ$ &  $1$ & $1$ & $1$ & $1$ & $1$ & $1$ & $1$ & $1$ \\
$YZY$ &  $1$ & $1$ & $1$ & $0$ & $1$ & $1$ & $0$ & $1$ \\
$ZYY$ &  $1$ & $1$ & $1$ & $1$ & $1$ & $0$ & $0$ & $1$ \\
$YZZ$ &  $0$ & $1$ & $1$ & $1$ & $0$ & $1$ & $1$ & $1$ \\
$ZYZ$ &  $0$ & $1$ & $0$ & $1$ & $1$ & $1$ & $1$ & $1$  \\
$ZZY$ &  $1$ & $1$ & $1$ & $1$ & $1$ & $1$ & $1$ & $1$ \\
$ZZZ$ & $0$ & $1$ & $0$ & $1$ & $0$ & $1$ & $1$ & $0$ \\
\end{tabular}
\end{center}
Note that this table can be obtained from the $AND$ table by simply relabeling the columns. The relabeling sends the first $+$ to $+$, the second $+$ to $+$ and the third $+$ to $-$, and it sends the first two $-$s to $-$ and the third one to $+$. 

The same argument used in the proof of Theorem \ref{and} can therefore be used to prove the logical contextuality of the $NAND$ state, with the provision that the new labeling replaces the one used within the old argument's statements.\qed
\end{proof}

\begin{remark}
The notation $+++\mapsto ++-$ unambiguously describes the relabeling used in the proof of Theorem \ref{nand}, and we shall use this shorthand notation in further proofs.
\end{remark}

\begin{theorem}
The $OR$, $NOR$, $L_1$, $NL_1$, $L_2$ and $NL_2$ states are all logically contextual.
\end{theorem}

\begin{proof}
The support of the probability tables for these states are also obtained from the $AND$ table by column relabelings, so the argument used in the proof of Theorem \ref{and} can again be used to prove the logical contextuality of these states. The necessary relabelings are

\begin{itemize}
                                                                                                                                                                                                                                                                                 \item[1)] $+++\mapsto ---$ for the $OR$ state
\item[2)] $+++\mapsto --+$ for $NOR$
\item[3)] $+++\mapsto +--$ for $L_1$
\item[4)] $+++\mapsto +-+$ for $NL_1$
\item[5)] $+++\mapsto -+-$ for $L_2$ 
\item[6)] $+++\mapsto -++$ for $NL_2$
                                                                                                                                                                                                                                                                                \end{itemize}
\qed
\end{proof}

\begin{remark}
The relabelings above can also be used for the probability tables themselves, not only for their supports, but only for $Y$, $Z$ measurements. For general choices of measurements there is no simple relation between the probability tables of the balanced states with functional dependency given by degree two polynomials, nor between their supports.
\end{remark}

\subsection{The $n+1$-partite case for $n>2$}

We can use the results of the previous section to classify the $n+1$-partite balanced states which have a functional dependency. In the rest of this section, let $F_n$ denote a polynomial in $n$ variables. 

\subsubsection{Strongly contextual states}

\begin{theorem}
Given a $n+1$-partite balanced quantum state whose functional dependency is given by the polynomial $F_n(q_1,\ldots,q_n)$, the state is strongly contextual if the polynomial $F_n$ is of the form
$$F_n(q_1,\ldots,q_n)=q_i+q_j+F_{n-2}(q_1,\ldots,\hat{q_i},\ldots,\hat{q_j},\ldots,q_n)$$
for some variables $q_i$ and $q_j$ and some polynomial $F_{n_2}$. 
\end{theorem}

\begin{proof}
If $Y$ and $Z$ measurements are chosen by each party, then we can show that none of the sections in the support of the $ZZZ\ldots Z$ row can be extended to a global section. 

Consider any fixed assignment of outcomes to the $Z$ measurements performed by the first $n$ parties, except the $i^{th}$ and the $j^{th}$ party. Let $\sigma_k\in\{+,-\}$, $k\neq i,j$ denote the outcome corresponding to the measurement performed by the $k^{th}$ party. Next evaluate the polynomial $F_{n-2}$ at the values of $q_1,\ldots,\hat{q_i},\ldots,\hat{q_j},\ldots,q_n$ corresponding to the fixed assignment of outcomes, using the convention that $0$ corresponds to the $+$ outcome and $1$ corresponds to the $-$ outcome. Use $a$ to denote the result of the evaluation.

Depending on the value of $a$ we can use the argument made for the strong contextuality of either the $XOR$ or the $NXOR$ state in order to show that there is no consistent assignment of outcomes which will restrict to sections in the support for all four of the following rows:

\begin{align*}
 Z\ldots ZZ_iZ\ldots ZZ_jZ\ldots ZZ\\
 Z\ldots ZY_iZ\ldots ZY_jZ\ldots ZZ\\
 Z\ldots ZZ_iZ\ldots ZY_jZ\ldots ZY\\
 Z\ldots ZY_iZ\ldots ZZ_jZ\ldots ZY
\end{align*}

As in the three-partite case, the contradiction comes from the fact that, depending on the value of $a$ and the outcome assigned to the $Z$ measurement of the $n^{th}$ party, the outcomes of the $Y_i$ and $Y_j$ measurements must be assigned equal values on the one hand, in order to be in the support of the $Z\ldots ZY_iZ\ldots ZY_jZ\ldots ZZ$ row, but on the other hand, opposite values in order to be in the support of the last two rows considered above, or viceversa. 

Since this can be done for all possible assignments of outcomes to the $Z$ measurements performed by the first $n$ parties, the quantum state we are considering must be strongly contextual.\qed
\end{proof}

\subsubsection{Logically contextual states}

\begin{theorem}
Any $n+1$-partite balanced quantum state whose functional dependency is given by a polynomial $F_n(q_1,\ldots,q_n)$ of degree at least two which is not of the form
$$F_n(q_1,\ldots,q_n)=q_i+q_j+F_{n-2}(q_1,\ldots,\hat{q_i},\ldots,\hat{q_j},\ldots,q_n)$$
for any choice of variables $q_i$ and $q_j$ and polynomial $F_{n_2}$ is logically contextual.
\end{theorem}

\begin{proof}
Consider any two variables $q_i$ and $q_j$ which appear in at least one of the terms with degree at least two of the polynomial $F_n$. The polynomial $F_n$ can be rewritten as 
$$F_n(q_1,\ldots,q_n)=F^1_{n-2}+q_iF^2_{n-2}+q_jF^3_{n-2}+q_iq_jF^4_{n-2}$$
where $F^i_{n-2}$ are $n-2$ variable polynomials in $q1,\ldots,\hat{q_i},\ldots,\hat{q_j},\ldots,q_n$.

Next choose any assignment of outcomes to the $Z$ measurements performed by the first $n$ parties, except the $i^{th}$ and the $j^{th}$ party, such that the polynomial $F^4$ evaluates to $1$ at the values of $q_1,\ldots,\hat{q_i},\ldots,\hat{q_j},\ldots,q_n$ corresponding to this assignment. Using this assignment, we have obtained a degree two polynomial in two variables, $q_i$ and $q_j$.

We can now use one of the arguments in Section \ref{two} in order to identify at least two sections in the support of the $$Z\ldots ZY_iZ\ldots ZY_jZ\ldots ZZ$$ row which cannot be extended to a global section consistent with the support of the rows

\begin{align*}
 Z\ldots ZZ_iZ&\ldots ZY_jZ\ldots ZY\\
 Z\ldots ZY_iZ&\ldots ZZ_jZ\ldots ZY\\
&\text{and}\\
 Z\ldots ZZ_iZ&\ldots ZZ_jZ\ldots ZZ \ \ \qed
\end{align*}  
\end{proof}

Note however that showing that at least one global section \textit{does} exist for the class of states considered in the Theorem above is not as simple as in the three partite case, so strong contextuality cannot be immediately ruled out for these states even in the special case when one considers only $Y$ and $Z$ measurements.

\subsubsection{Weakly contextual states}

\begin{theorem}
Any $n+1$-partite balanced quantum state whose functional dependency is given by a polynomial $F_n(q_1,\ldots,q_n)$ of degree one which is not of the form
$$F_n(q_1,\ldots,q_n)=q_i+q_j+F_{n-2}(q_1,\ldots,\hat{q_i},\ldots,\hat{q_j},\ldots,q_n)$$
for any choice of variables $q_i$ and $q_j$ and polynomial $F_{n_2}$ is weakly contextual.
\end{theorem}

\begin{proof}
Any degree one polynomial which is not of the above form must contain precisely one term. Thus the state we are dealing with is a so-called dictatorship state, i.e. the value of the last qubit is dictated by the value of its $i^{th}$ qubit, and the state is either of the form
$$\Delta^+_i:=\left(\frac{|0\rangle+|1\rangle}{\sqrt{2}}\right)^{\otimes n}\otimes \frac{|0_i0_{n+1}\rangle+|1_i1_{n+1}\rangle}{\sqrt{2}}$$
or
$$\Delta^-_i:=\left(\frac{|0\rangle+|1\rangle}{\sqrt{2}}\right)^{\otimes n}\otimes \frac{|0_i1_{n+1}\rangle+|1_i0_{n+1}\rangle}{\sqrt{2}}$$
and its probability table can be expressed in terms of a suitable choice of $n-2$ constants and the probability table of either the $\Phi^+$ or of the $\Phi^-$ state.

A straightforward inductive argument based on the argument used in Proposition \ref{bell} will show that the $n+1$-partite dictatorship states are also weakly contextual for the measurements $U \left(\frac{\pi}{2},\frac{\pi}{8}\right)$, $U\left(\frac{\pi}{2},\frac{5\pi}{8}\right)$ and $U\left(\frac{\pi}{8},\frac{\pi}{2}\right)$, $U\left(\frac{5\pi}{8},\frac{\pi}{2}\right)$ respectively. 

Moreover, the generalization of the argument used in Theorem \ref{bellth} shows that the $n+1$-partite dictatorship states are not logically contextual for any possible dichotomic choice of measurements.\qed 
\end{proof}

\subsubsection{Non-contextual states}

Any $n+1$-partite balanced quantum state whose functional dependency is given by a constant polynomial is clearly a product state and hence non-contextual.

\section{Final Remarks}

In this paper, we have shown the logical contextuality of two classes of states, the permutation-symmetric and functionally dependent states.
Our proofs have been constructive, in that we have explicitly given local observables which witness the logical contextuality of these classes of states.

What about the general case?
In the forthcoming paper \cite{ACY} with Shenggang Ying, we establish the following result.

Let $P(n)$ be the class of $n$-qubit pure states which, up to permutation, can be written as tensor products of 1-qubit and 2-qubit maximally entangled states. Let $L(n)$ be the set of logically contextual $n$-qubit states.

\begin{theorem}
For all $n \geq 1$, $P(n)$ and $L(n)$ partition the set of $n$-qubit pure states.
\end{theorem}

Thus every pure state is \emph{either} a state whose only form of entanglement is bipartite maximal entanglement in 2-qubit subsystems; \emph{or} it is logically contextual.
So logical contextuality, with certain bipartite exceptions, holds in general.

%By contrast with the results proved in the present paper, this theorem is proved non-constructively. It remains to be seen if a more direct argument is possible.

This result can moreover be proved constructively, leading to an algorithm which, given an $n$-qubit state, either returns that it is in $P(n)$, or produces local observables which witness the logical contextuality of the state. Strikingly, only $n+2$ local observables are needed for a $n$-qubit state.

\bibliographystyle{eptcs}

%\bibliography{bdbib}
\end{document}